\documentclass[11pt,a4paper,english]{article} 
\usepackage{authblk}
\title{Curving Flat Space-Time by Deformation Quantization?}
\author{Albert Much\footnote{amuch@matmor.unam.mx}}
\affil{Centro de Ciencias Matem\'aticas,\\
	Universidad Nacional Aut\'onoma de M\'exico,\\
	C.P. 58190, Morelia, Michoac\'an, Mexico
} 
\usepackage{float}            
\usepackage[T1]{fontenc}   
\usepackage[cmex10]{amsmath}  
\usepackage{amstext}          
\usepackage{mathrsfs}
\usepackage{amsfonts}         
\usepackage{amssymb}          
\usepackage{amsthm}    
\usepackage{setspace}
\usepackage{bm}               
\usepackage{enumerate}        
\usepackage[bottom]{footmisc} 
\usepackage{array}            
\usepackage{textcomp}         
\usepackage{pdfpages}         
\usepackage{parskip}          
\usepackage[right]{eurosym}   
\usepackage{xcolor}           
\usepackage[square,numbers]{natbib}	  

\usepackage{amsmath}  
\usepackage[colorlinks,pdfpagelabels,pdfstartview = FitH,bookmarksopen = true,bookmarksnumbered =
true,linkcolor = black,plainpages = false,hypertexnames = false,citecolor = black]{hyperref}

\newtheorem{theorem}{\textsc{Theorem}}[section]
\newtheorem{lemma}{\textsc{Lemma}}[section]
\newtheorem{proposition}{\textsc{Proposition}}[section]

\theoremstyle{definition}
\newtheorem{definition}{\textsc{Definition}}[section]

\newtheorem{convention}{Conventions}[section]

\theoremstyle{remark}
\newtheorem{remark}{Remark}[section]

\usepackage{fancyhdr}
\numberwithin{equation}{section} 
\begin{document}
	\maketitle
	\abstract{We use a deformed differential structure  to obtain a curved metric by a   deformation quantization of the  flat space-time. In particular, by setting the deformation parameters to be equal to physical constants we obtain  the Friedmann-Robertson-Walker (FRW) model for inflation and a deformed version of the FRW space-time. By calculating classical Einstein-equations for the extended space-time we obtain non-trivial solutions. Moreover, in this framework we obtain the Moyal-Weyl, i.e. a constant non-commutative space-time, as a consistency condition.} 
	\tableofcontents
 
	\section{Introduction}
	In the search for a complete   theory of quantum gravity there has been many proposals over the last years. One approach that has gained quite some popularity comes from mathematics and is known by non-commutative geometry. The main motivation to study non-commutative geometry in the realm of physics comes from the "geometrical" measurement problem, \cite{DFR}. Basically, the problem goes as follows; by combining principles of quantum mechanics and general relativity it turns out that the measurement of a space-time point with arbitrary precision is not possible and thus space-time, around the Planck length, does not have a continuous structure. Hence, geometry of space-time has to be replaced by  a non-commutative version thereof. One of the most studied and well understood examples of such a non-commutative geometric structure is known by the Moyal-Weyl plane. In essence, one has a constant non-commutativity between the space-time coordinates (that are replaced by operators). This is equivalent to the non-commutativity that is introduced between the observables, i.e. momentum and coordinate, in quantum mechanics. 
	\\\\
	Another essential reason to study non-commutative geometry is the search for quantum gravity.  Let us elaborate on this point a bit further.   Essentially, the Einstein Field Equations of general relativity tell us that gravity is a force,   experienced by the curvature of space-time. Hence,  ultimately quantization of gravity corresponds to the  quantization  of space-time.  What is quantization? This question has many possible answers. However, all answers have something in common. Namely, we take a classical theory and by introducing, for example, a new product, we extend  the classical framework. An interesting example is to take classical mechanics and perform the so-called \textbf{deformation quantization} in order to obtain quantum mechanics but in terms of functions with a non-commuting product, see \cite{SZ7}. Hence, in the former example we went from classical mechanics to quantum mechanics, by replacing the point-wise product of the phase-space coordinates by a non-commutative (but associative) product. Following the same train of thought,   deformation quantization of the point-wise product of 	classical space-time coordinates  leads us to space-time coordinates that multiply with a non-commutative product,   \cite{SZ1}. Therefore, by replacing the commutative point-wise product on classical space-times by a non-commutative one, we turn a classical space-time into a non-commutative space-time. 	Since deformation   takes us   from classical to quantum we call the resulting non-commutative space-times,   quantum space-times. \\\\
	However, the deformation quantization map from classical to quantum is neither unique nor the only path to quantize a theory. Furthermore, non-commutative algebras do not necessarily lead to non-commutative space-times. In particular, non-commutative algebras can as well be associated with curvature of the space-time.\footnote{We are indebted to the referee for this remark.} Let us mention the example of the Anti-de Sitter space, see \cite{SZ2} for a recent review. The deformation of the commutative algebra of  translations in Minkowski space-time leads to a non-commutative structure between the translations. The deformation parameter dictating this non-commutative structure is interestingly enough given by the cosmological constant. Hence, by the deformation of the translational algebra (that is a sub-algebra of the Poincar\'e algebra) the flat space-time (Minkowski) becomes the curved space-time known by Anti-de Sitter. Note, that not only deformation quantization maps the flat to the curved case but also deformation in a Lie-group contextual sense, can achieve that. Basically stability considerations, in particular the vanishing of the second co-homology group can achieve the before-mentioned outcome (see \cite{SZ3}).
	\\\\ 
	Another, equally valid approach to quantize a theory, be it quantum mechanics, i.e. the phase space of classical mechanics,	 or  space-time is to start with a non-commutative algebra given by   operators on a Hilbert space. Hence, the so called quantization map  promotes  the classical functions to operators, see \cite{GG}  and \cite{GO}.  With regards to non-commutative spacetimes, recent success has been achieved in this direction   by obtaining, through techniques developed and conditions asserted in non-commutative geometry, the equivalent of a metric,    see \cite{BM}. 	This approach differs from the ones described in the former paragraphs, since here one deals with non-commutating operators instead of functions with a non-commuting product. 	\\\\  In this work we focus on this approach in order to be as near to the standard approach to quantum mechanics as possible, see \cite{GA1}. 	 Furthermore, the idea in this article is to begin with the Minkowski metric, that is given by self-adjoint operators, and by a strict deformation quantization induce a curved space-time metric.      Hence, the concrete question that we pose is the following: Can we obtain a curved space-time from a purely flat space-time by strict deformation quantization of linear operators that are defined on a Hilbert space? Hence, is it possible to obtain gravity in this framework? Why is that of particular interest? 		Intuitively,  there are three main ideas that motivate this work. The first idea takes into account that in the standard approach of quantum mechanics  observables are given by self-adjoint unbounded operators on a Hilbert-space.   Hence,   in a mathematically well-defined theory of quantum gravity,  a quantity   identified with the metric,   has to be given in terms of self-adjoint operators  defined on a dense subset of the Hilbert space as well. 	The second intuitive idea is the fact that some quantum gravity theories, such as Loop quantum gravity, see \cite{SZ4} for a recent review, face the issue of obtaining classical general relativity in the classical limit let alone the flat space-time (see \cite{SZ5}). Therefore we start with the flat metric. Third,    mathematically  well-defined considerations, such as strict deformation quantization,  should lead us from   flat  to  curved space-time, hence obtaining  gravitational effects already in terms of operator-valued terms.       
	\\\\ 
Since, a scheme of obtaining a curved  metric from a flat one by a strict deformation quantization procedure is still missing, this paper intends to resolve  this issue. We present a concrete and moreover (mathematically) strict scheme in which it is possible  to apply a deformation quantization of  the flat metric and obtain a curved space-time. The emergence of curvature by a deformation quantization is achieved by combining two major mathematical developments in non-commutative geometry. The first development that we use is the universal differential structure of Connes \cite[Chapter 3, Section 1]{C}
that associates to any associative unital algebra a differential structure. While the second framework used  is the Rieffel deformation, \cite{R} (and extensions thereof  see \cite{GL1}, \cite{GL2}, \cite{BS}, \cite{BLS}), that deals with strict deformation quantizations of $C^{*}-$algebras. 
	\\\\
	Let us explain in further detail how those two frameworks are combined.	 The initial point of our considerations is the deformation of the standard 	(commutative) differential structure (see \cite{DHS} or/and \cite{BM}), i.e.
	\begin{equation*}  
	[\hat{x}^{\mu},d\hat{x}^{\nu}]=\sum\limits_{\sigma=0}^{n} C^{\mu\nu}_{\,\,\,\,\sigma}\,d\hat{x}^{\sigma},
	\end{equation*}
where the constants $C$ are imaginary valued and symmetric in the first two upper indices. This deformation   is a solution of relations obtained by the application  of a linear operator, satisfying the Leibniz rule, on the commutation relations of commuting  coordinate operators. Further on, we define the flat metric as the tensor product of two differentials and consider deformations of the flat  line-element. This particular line-element is defined by using the flat metric.  Considering deformations of the line-element instead of the metric itself  has, apart from technical reasons, the   physical intuition that emphasizes the importance on the line-element. In particular, the deformations are performed with the so-called warped convolutions \cite{BLS} which supply isometric representations of Rieffel's strict deformations of $C^*$-dynamical systems. In order to use the deformation procedures we have to find representations of this algebra. We were able to achieve this by using the richness of the Heisenberg-Weyl algebra. Since  we work with unbounded representations we use the seminal work of \cite{P} to remain in a rigorous framework. Before analyzing the outcome of the deformation  we prove the rigorousity of the scheme. In particular, we prove the convergence of the oscillatory integrals defining the deformation,  self-adjointness and  uniqueness of the deformed operators.
\\\\
The outcome of this deformation quantization is interesting. We obtain a coordinate dependent conformal transformation of the flat undeformed metric. Hence, we obtain a whole family of conformal flat metrics. Furthermore, we investigate the outcome of such a scheme  in a physical context. It turns out that, for a particular choice of the deformation parameters of the algebra, we obtain a deformed Friedmann-Robertson-Walker (FRW) space-time. Moreover, in a certain limit we obtain the FRW space-time for a dark energy dominated universe. Another example of a physical interesting  space-time is obtained by this scheme, namely that of an  ultra-static one. In addition to curving space-time by a strict deformation quantization of the flat space-time metric we obtain the necessity of a non-commutative space-time, namely that of the Moyal-Weyl plane, purely out of a consistency requirement. We do not claim that this holds in general, i.e. for all curved space-times that are given by operator-valued objects. However, for the class of conformal flat metrics that we obtained by deformation quantization, that in addition are given as operator-valued tensor products, for centrality of the metric tensor to hold (see \cite{BM}) the commutation relations of the Moyal-Weyl plane have to be imposed on the algebra of the space-time operators. \\\\
The paper is organized as follows; In the second section we give a compact summary of the non-commutative differential calculus, representations of unbounded algebras and the deformation related to the Rieffel deformations known as warped convolutions. The third section examines the well-definedness, self-adjointness and uniqueness  of the deformation. In section four we turn our attention to gives examples of physically relevant space-times obtained by this scheme.  Section five investigates the outcome  of    non-commutativity between space and time by demanding centrality of the metric tensor. We end this work with a conclusion and an outlook on possible extensions of the deformation. 
	\begin{convention}
		Throughout this work we use $d=n+1$, for $n\in\mathbb{N}$. The Greek letters are split into  $\mu, \,\nu=0,\dots,n$. We use Latin letters  $i,\,j,\,k,\cdots$ for  spatial components which run from $1,\dots,n$ and we let the letters $b,\,c$ run from $1,\dots,d$. Furthermore,  we choose the following convention for the Minkowski scalar product of $d$-dimensional vectors, $a\cdot p=a_0p^0+a_kp^k=a_0b^0- \vec{a}\cdot\vec{p}$. Moreover, we use the following notation concerning inner products of matrices with vectors, $(\Theta x)_{\mu}=\Theta_{\mu\nu} x^{\nu}$.
	\end{convention}
	\section{Preliminaries}
	This section is separated in two parts. The first part gives a swift introduction of the non-commutative differential calculus. Although this calculus can already be viewed as a deformation of the commutative case, we intend to perform another deformation.  In particular, the quantization we use in this context is  known by the name of Rieffel \cite{R} deformation or as an   extension  thereof, namely warped convolutions \cite{BLS}. These strict deformations make extensive use of the spectral calculus. Hence, in order to apply these quantizations on our non-commutative differential calculus we need to have representations of the respectively chosen algebras on a Hilbert space. Primarily, we need to have self-adjoint representations of the $^{*}$-algebras in order to use warped convolutions.  	Therefore, the second part of this section is devoted to  the connection between the non-commutative differential calculus and representations of $^{*}$-algebras. 
	\subsection{Non-commutative Differential Calculus}
	In order to give a thorough path to non-commutative geometry we define in the next step the universal differential calculus of an arbitrary associative algebra (not necessarily commutative)  as follows (see \cite[Chapter 3, Section 1]{C} or \cite[Chapter 6, Section 1]{L}).
	\begin{definition} \textit{Universal differential algebra}\\\\
		Let $\mathcal{A}$ be an associative algebra with unit over $\mathbb{C}$. 	Then, the \textbf{universal differential algebra of forms}  $\Omega(\mathcal{A})=\bigoplus_{p }  \Omega^{p}(\mathcal{A})$ is a graded algebra defined as follows. For $p=0$ we have $\Omega^{0}(\mathcal{A})=\mathcal{A}$. The space $\Omega^1(\mathcal{A})$ of one-forms is generated, as a left $\mathcal{A}$-module by a  $\mathbb{C}$-linear   operator $d:\mathcal{A}\rightarrow\Omega^1(\mathcal{A})$, called the universal differential, which satisfies the relations, 
		\begin{equation} \label{lr}
		d^2=0,\qquad	d(a_1a_2)=(da_1)a_2+a_1da_2, \qquad \forall a_1,\,a_2\in  \mathcal{A} .
		\end{equation}
		The left $\mathcal{A}$-module $\Omega^1(\mathcal{A})$ can be endowed  with a structure of a right 
		$\mathcal{A}$-module as well by   using the Leibniz rule given in Equation (\ref{lr}). Hence $\Omega^1(\mathcal{A})$ is a bimodule.  Moreover, the space $\Omega^{p}(\mathcal{A})$ is defined as
			\begin{equation*}  
		 \Omega^{p}(\mathcal{A})=\underbrace{\Omega^1(\mathcal{A})\cdots \Omega^1(\mathcal{A})}_{p-times}.
			\end{equation*}
	\end{definition}
Next, we turn our attention to a concrete algebra given in	\cite{DHS}. The authors start with  a commutative algebra $\mathcal{A}_c$ which is generated by self-adjoint commutative coordinate operators $\hat{x}^{\mu}$  fulfilling, 
	\begin{equation}\label{cr1}                     
	[\hat{x}^{\mu},\hat{x}^{\nu}]=0.
	\end{equation}
By applying the operator $d$ on the commutator relations of the algebra $\mathcal{A}_c$ (Equation (\ref{cr1})), we have
	\begin{align}\label{cr3}
	[d\hat{x}^{\mu},\hat{x}^{\nu}]+ [\hat{x}^{\mu},d\hat{x}^{\nu}]=0
.\end{align}
There are two possible solutions to these commutator equations. The first one is for the commutators to be separately zero, which in turn corresponds to the commutative case. However, a more general solution is given by  
	\begin{equation} \label{cr2}
	[\hat{x}^{\mu},d\hat{x}^{\nu}]=\sum\limits_{\sigma=0}^{n} C^{\mu\nu}_{\,\,\,\,\sigma}\,d\hat{x}^{\sigma},
	\end{equation}
	where here the sum was written to indicate that the index $\sigma$ does not follow the Einstein convention, i.e. it is   the usual (Euclidean) sum.

	\begin{lemma}\label{l1}
		The constants $C^{\mu\nu}_{\,\,\,\,\sigma}$ have to  fulfill the following  consistency conditions. \\
		\begin{enumerate} \renewcommand{\labelenumi}{(\roman{enumi})}
			\item 	 $C^{\mu\nu}_{\,\,\,\,\sigma}$ are symmetric in the first two indices,  
			\begin{align*}
			C^{\mu\nu}_{\,\,\,\,\sigma}= C^{\nu\mu}_{\,\,\,\,\sigma}.
			\end{align*}
			\item By demanding a consistent differential calculus,   we have the following restrictions on the constants $C$,  
			\begin{align*}
			\sum\limits_{\sigma=0}^{n} \left(C^{\mu\nu}_{\,\,\,\,\sigma} 	 C^{\lambda\sigma}_{\,\,\,\,\kappa}-C^{\lambda\nu}_{\,\,\,\,\sigma} 	 C^{\mu\sigma}_{\,\,\,\,\kappa}\right)  =0.
			\end{align*}
	 
		\end{enumerate}
	\end{lemma}
	\begin{proof}$\,$\\
		\begin{enumerate} \renewcommand{\labelenumi}{(\roman{enumi})}
			\item The first point can be seen directly from the defining Equation (\ref{cr3}).\\
			\item We prove the second item by using the commutativity of the coordinate operators $\hat{x}$ and by explicitly calculating  	 the Jacobi identities for    Equation (\ref{cr2}). The identities are  given by,
			\begin{align*}
	 	[ \hat{x}^{\lambda}, [\hat{x}^{\mu},d\hat{x}^{\nu}]]+[ \hat{x}^{\mu}, [d\hat{x}^{\nu},\hat{x}^{\lambda}]]  &=
		\sum\limits_{\sigma=0}^{n}\left( C^{\mu\nu}_{\,\,\,\,\sigma}\, 	[ \hat{x}^{\lambda},d\hat{x}^{\sigma}]
-C^{\lambda\nu}_{\,\,\,\,\sigma}\, 	[ \hat{x}^{\mu},d\hat{x}^{\sigma}]	\right)	\\&=
		\sum\limits_{\sigma=0}^{n} \sum\limits_{\kappa=0}^{n}\left(C^{\mu\nu}_{\,\,\,\,\sigma} 	 C^{\lambda\sigma}_{\,\,\,\,\kappa}-C^{\lambda\nu}_{\,\,\,\,\sigma} 	 C^{\mu\sigma}_{\,\,\,\,\kappa}\right)	d\hat{x}^{\kappa}  .
			\end{align*}

		\end{enumerate}
	\end{proof} 
	Since we demand, in the following sections, that the constants $C$ are imaginary we can extend the conjugation $^{*}$ of complex numbers to the differential algebra $\Omega(\mathcal{A}_c)$. In addition to demanding that $^{*}$ commutes with $d$ and that $^{**}$ is the identity we have the assumption that $^{*}$ is a graded anti-automorphism (see \cite{DHS}).  Next, in order to represent our algebra $\mathcal{A}_c$ and its respective differential calculus on a Hilbert space, we first  prove that we have a $^*$-algebra at hand. 
	Therefore, let us define in a short manner a general $^*$-algebra, \cite{P}.
	\begin{definition}
		A $^*$-algebra $\mathcal{A}$ is an algebra over the complex numbers with a $^*$-operator satisfying;
		\begin{enumerate} \renewcommand{\labelenumi}{(\roman{enumi})}
			\item  $A^{**}=A$,
			\item $(\alpha A+ B)^*=\overline{\alpha}A^*+B^*$
			\item $(AB)^*=B^*A^*$	,
		\end{enumerate}
		for all $A,B\in\mathcal{A}$ and complex numbers $\alpha$.  Since, we work with representations of unbounded algebras in general our $\mathcal{A}$ will not be normed. 
	\end{definition}
	\begin{lemma}
		The algebra $\mathcal{A}_c$  is a $^*$-algebra. 
	\end{lemma}
	\begin{proof}
		Due to   the assumption that the graduation is acting as an anti-automorphism it follows that the algebra $\mathcal{A}_c$  $^*$-algebra. 
	\end{proof}

	\subsection{Representation of  the Deformed Differential-Structure}
	
	In this section we want to give representations of the algebra  $\mathcal{A}_c$ and its universal differential. To do so, we use the \textbf{Schr\"odinger 
		representation} \cite[Chapter IIIV.5]{RS1} of the Heisenberg-Weyl algebra. In this context the  operators
	$(Q^{b},P^{c})$, satisfying the \textbf{canonical commutation relations}  (CCR)
	\begin{equation}\label{ccr}
	[Q^b,P^{c}]= i\delta^{bc},
	\end{equation}
	are represented as   self-adjoint operators on the Schwartz space
	$\mathscr{S}(\mathbb{R}^d)$. Here $Q^{b}$ and $P^{c}$ are the closures of
	$q^{b}$ and multiplication by $i {\partial}/{\partial q_c}$ on the dense domain
	$\mathscr{S}(\mathbb{R}^d)$ respectively. In the next step we define a representation of a $^*$-algebra that is not necessarily bounded, \cite{P}.
	
	\begin{definition}\label{defsr}
		A representation $\pi$ of an algebra $\mathcal{A}$ on a Hilbert space $\mathscr{H}$ is a mapping of $\mathcal{A}$ into linear operators all defined on a common dense domain $\mathcal{D}(\pi)$ $\pi:\mathcal{A}\rightarrow \mathcal{D}(\pi)$ and $\pi$ satisfies $\pi(1)=\mathbb{I}
		_{\mathscr{H}}$:\\
		\begin{enumerate} \renewcommand{\labelenumi}{(\roman{enumi})}
			\item  $\pi(\alpha A+B)f=\alpha\pi( A)+\pi(B)f$, for all  $A,B\in\mathcal{A}, f\in  \mathcal{D}(\pi)$ and all $\alpha\in\mathbb{C}$.
			\item 	$\pi(A)\mathcal{D}(\pi)\subset\mathcal{D}(\pi)$ for all $A\in\mathcal{A}$ and \\$\pi(A)\pi(B)f=\pi(AB)f$ for all $A,B\in\mathcal{A}$ and $f\in \mathcal{D}(\pi)$.
			\item  A representation $\pi$ of a $^*$-algebra $\mathcal{A}$ on a Hilbert space is said to be \textbf{hermitian} or a \textbf{$^*$-representation} if \\
			$\left(f,\pi(A)g\right)=\left(\pi(A^*)f,g\right)$ for all $f,g\in \mathcal{D}(\pi)$ and $A\in\mathscr{H}$ i.e. $\pi(A^*)\subset\pi(A)^*$.
		\end{enumerate}$\,$\\
		A representation $\pi$ is \textbf{hermitian} iff for every hermitian $A\in\mathcal{A}$, $\pi(A)$ is hermitian. Furthermore 	a representation $\pi$ is \textbf{self-adjoint} iff, in addition to hermiticity, $\mathcal{D}(\pi)=\mathcal{D}(\pi^*)$.
	\end{definition} 

	 Next, we give a representation of the algebra $\mathcal{A}_c$ and prove that it corresponds to a $^*$-representation. 
	\begin{lemma}
		Let the $d$-dimensional algebra $\mathcal{A}_c$ be given as,
		\begin{align*}
		[\hat{x}^{\mu},\hat{x}^{\nu}] =0.
		\end{align*}   Then, a faithful $^*$-representation of this  algebra, denoted by  $\pi: \mathcal{A}_c \rightarrow \mathscr{H}=\mathscr{L}^2(\mathbb{R}^d)$, is given in terms of unbounded self-adjoint operators on the dense domain  $\mathcal{D}(\pi)=\mathscr{S}(\mathbb{R}^d)$ by
		\begin{align}\label{al1}
		\pi(\hat{x}^{\mu})&=\frac{a^{\mu}}{2}(Q^{\mu+1}P^{\mu+1}+P^{\mu+1}Q^{\mu+1}),
				\end{align} 
where $a^{\mu}\neq0$ is a real vector.
	\end{lemma}
	
	\begin{proof}
	First, we prove that the  representation $\pi$ satisfies all the conditions required from the definition of such a representation (see Definition \ref{defsr}). The requirement of linearity  is straight forward. The first part of the second condition, i.e. 
		$\pi(\hat{x})\mathcal{D}(\pi)\subset\mathcal{D}(\pi)$ for all $\hat{x}\in\mathcal{A}_c$, follows from the fact that we  work with the Schr\"odinger representation of the Heisenberg-Weyl algebra. Since arbitrary polynomials of the representation of this algebra have the stable domain $\mathscr{S}(\mathbb{R}^d)$, the proof is completed. The multiplicative  and the $^*$-representation  property follows  from the Schr\"odinger representation, i.e. 	since the representations of the algebras are hermitian polynomials of the Heisenberg algebra it follows from \cite{P}[Section V. Example 2] that the representations are self-adjoint.   Moreover, if the vector $a^{\mu}\neq0$ the $^*$-representation $\pi$ is faithful, i.e. $ker(\pi)=\{\emptyset\}$ and hence an inverse exists. Since, by construction,  the symmetric operator $\pi(\hat{x}^{\mu})$ is the dilatation operator 	self-adjointness follows, \cite[Equation 10.9]{T}.
	\end{proof} 
	 Next we give the representation of one-forms of our respective algebra. It is in the following defined as a commutator of the momentum operator with the representation of elements of the algebra. The reason therein lies in the fact that we need  a representation of the exterior derivative that fulfills the conditions required from such a $\mathbb{C}$-linear operator.
		
\begin{lemma}\label{ex1}
Let a consistent differential algebra be given as 
	 \begin{align}\label{rd1}
	 [ \hat{x}^{\mu} , d\hat{x}^{\nu} ]  = ia^{\mu}\delta^{\mu\nu} d\hat{x}^{\nu} .
	 \end{align}
	Then,  a representation of the
	universal differential as a  derivation
	 can be defined by the following commutator relation, 
\begin{align}\label{rd}
\pi(d\hat{x}^{\mu}):= i{q}_{b}\,[P^{b} , \pi(\hat{x}^{\mu})],
\end{align}
where  $q$ is a real non-zero vector. The representation of the derivative obeys Leibniz and satisfies 
\begin{align} 
 \pi(d^2\hat{x}^{\mu})=i{q}_{b}\,[P^{b}, \pi(d\hat{x}^{\mu})]= -{q}_{b}\,{q}_{c}\,[P^{b},[P^{c},\pi( \hat{x}^{\mu})]]=0 .
 \end{align}
	\end{lemma} 

\begin{proof}
 	From the Jacobi identities of the deformed differential algebra (see Lemma \ref{l1}) it is easily verified that the given algebra in this lemma is consistent (see as well \cite{DHS}).   Moreover, the fact that the derivative satisfies the required assumptions  of a   $\mathbb{C}$-linear operator follows from the commutator properties. In particular, nilpotency of the operator follows from the fact that the representations of the algebra $\pi(\hat{x})$ is linear in the coordinate. Hence, by applying the commutator of the momentum twice, the resulting object vanishes. Furthermore, the representation of the differential is self-adjoint due to the self-adjointness of the momentum operator and the representation of the algebra. 
	\end{proof}	  
	\begin{remark}
Let us mention that	the representations of the universal differential  by a commutator    have the same spirit as Connes' differential forms, \cite[Chapter 4]{C} (or see \cite[Chapter 6.2]{L}).
	\end{remark}
	
	\section{Deformation of the Line-Element}
	The main idea of this investigation is to obtain a curved metric  by a deformation of   flat space-time. There are different models since the deformation method of warped convolutions and the choice of a particular algebra of the differential structure allows an extensive family of toy models.  
	However, in this section we restrict ourselves to the simplest algebra. Next, we define the flat metric $\eta$ as the tensor product of two one-forms, i.e. $\eta\in \Omega^{1}(\mathcal{A}_c) \otimes_{\mathcal{A}_c} \Omega^{1}(\mathcal{A}_c) $
	\begin{align} 
	\eta=\eta_{\mu\nu}d\hat{x}^{\mu}\otimes_{\mathcal{A}_c} d\hat{x}^{\nu}.
	\end{align}
	It is more natural to use the $\mathcal{Z}(\mathcal{A}_c)$-module tensor product, where $\mathcal{Z}(\mathcal{A}_c)$ denotes the centre of the algebra as in \cite{BDF}. However, since our algebra is commutative the tensor product that we take agrees with the most natural choice.\\\\	We could directly deform the metric, however in physics an object of more interest is the    line-element. For the flat Minkowski case  it is given by 
		\begin{align} 
		ds^2=\eta_{\mu\nu}\, d\hat x^{\mu}d\hat x^{\nu}.
		\end{align}	
	Since this is the object of interest we intend to deform this entity and interpret the resulting deformation as a curvature of space-time. Before we deform the line-element we start with a proposition and  definition. 
	 \begin{proposition}\label{ur1}
	 	Let a $^*$-representation of the operators   $\hat{x}^0,\dots,\hat{x}^n$, denoted as $\pi(\hat{x}^0),\dots,\pi(\hat{x}^n)$,    be given by mutually commuting self-adjoint operators on the dense domain  $\mathcal{D}(\pi)\subseteq\mathscr{H}$. Then, there is a strongly continuous $(n+1)$-parameter group of unitary operators $U(p)$ on $\mathscr{H}  $ given as:
	 	\begin{equation}U(p)=\mathrm{exp}( {ip_{\mu}\pi(\hat{x}^{\mu}})),\qquad\forall p\in\mathbb{R}^{d}.
	 	\end{equation} 
	 	\begin{proof} 
	 		This is simply a restatement of Stone's theorem in the notation of $^*$-representations. 	\end{proof}	
	 \end{proposition}
	 	To deform flat space-time we use the method of warped convolution which supplies isomorphic representations of the so-called Rieffel deformations, \cite{BLS, GL1}. In the following we give the basic definition of the deformation. 
	 \begin{definition}\label{defwca}
	 	Let $\Theta$ be a skew-symmetric matrix on $\mathbb{R}^{d}$ and   $\mathcal{D}$ be the dense and  stable subspace of vectors in $\mathcal{H}$, which transform smoothly under the unitary operator $U$. Finally, let   $E$ be the spectral resolution of the   operator $U$. Then, the warped convolutions  of an operator $A$, defined on  $\mathcal{D}\subset\mathcal{H}$ and   denoted by $A_{\Theta}$, is  defined according to
	 	\begin{equation}\label{WC1}
	 	A_{\Theta} :=\int dE(x)\,\alpha_{\Theta x}(A) = (2\pi)^{-d}
	 	\lim_{\epsilon\rightarrow 0}
	 	\iint dx\, dy \,\chi(\epsilon x,\epsilon y )\,e^{-ixy}\, \alpha_{\Theta x}(A)\, U(y) ,
	 	\end{equation}
	 	where $\chi \in\mathscr{S}(\mathbb{R}^d\times\mathbb{R}^d)$ with $\chi(0,0)=1$ and $\alpha$ denotes the $\mathbb{R}^d$-action, i.e.
	 	$$
	 	\alpha_p(A):=U(p)\, A\,U(p)^{-1}\qquad  \forall p\in\mathbb{R}^{d}.
	 	$$
	 \end{definition}
	 The above definition has to be understood in a formal context. However, it is  rigorously defined  for a certain sub-algebra of bounded operators as was shown in 
	 \cite{BLS}. Since it is a fact of life that we have to deal with unbounded operators, in this context, we cannot use the proof of the bounded case. Since we, as well, strive for rigorousity we prove in the forthcoming sections the well-definedness of the oscillatory integrals for the unbounded case.
	\begin{definition}\label{defm1}
		Let the   deformed line-element, denoted by $ds^2_{\Theta}$, be defined as  
		\begin{align} 
	ds^2_{\Theta}:=	 (\eta_{\mu\nu}\,d\hat x^{\mu}d\hat x^{\nu})_{\Theta},
		\end{align}
		where the deformation is performed by using the unitary operators, given as in Proposition \ref{ur1}, and the deformed differentials (since a deformed constant, i.e. the flat metric, is simply equal to the undeformed case) are defined as, 
		\begin{align} 
		(d\hat{x}^{\mu}d\hat{x}^{\nu})_{\Theta}:=\pi^{-1}(\pi(d\hat{x}^{\mu}) \pi(d\hat{x}^{\nu}))_{\Theta}.
		\end{align}
	\end{definition}
	The reason for such a definition goes as follows. In order to perform the deformation we need the spectral representation. This is only given for self-adjoint representations. Hence, we are obliged to represent the algebra elements as self-adjoint (and unbounded) operators in order to proceed with the deformation. Therefore, we first take the representation of elements of the algebra and the differential structure, perform a deformation and map back to the abstract algebra of the universal differential algebra. Since the metric is  made up of a tensor product of the differentials deformation acts  solely on the differentials.

	\begin{remark}
		In order to ease readability we define the representations of the algebra and differential structure as follows,
		\begin{align*}
		\pi(\hat{x}^{\mu})=X^{\mu} , \qquad \pi(d\hat{x}^{\mu})=dX^{\mu} ,
		\end{align*}
		where the symbol $X$ should not be confused with the Schr\"odinger representation of the
		coordinate operator. 
	\end{remark}
	
	\subsection{The Simplest Case}
	In this section we deform the flat space-time in   order to obtain  curved geometry. The simplest deformation is done by using solely  the representations of the coordinate operators, i.e.  $X$, linearly in the algebra.   
\\\\
The particular algebra chosen in this section is the one of Lemma \ref{ex1}. Notice that this algebra is a very particular choice and in essence it can be extended (obeying the Jacobi identities) and therefore the study of deformations with an extended algebra can be studied. This in turn leads to more complex curved space-times. However, driven by  simplicity of the approach and a natural  physical outcome we stick to the simplest case. Since, we work with unbounded representations and not bounded operators  (as in \cite{BLS}) we have to prove that the  formula is well defined. In order to proceed we write 
the deformed operator in the scalar product     w.r.t. $\mathcal{H}$.
\begin{align*}
\langle \Psi,(dX^{\mu}dX^{\mu})_{{\Theta }}\Phi\rangle&=
\langle \Psi,(dX^{\mu})^2_{{\,\,\Theta }}\Phi\rangle \\&=
(2\pi)^{-d}
\lim_{\epsilon\rightarrow 0}
\iint  \, dx \,  dy \, e^{-ixy}  \, \chi(\epsilon x,\epsilon
y) {\langle \Psi,  
	U(y)\alpha_{\Theta x}(dX^{\mu})^2   \Phi\rangle} \nonumber \\&
=:
(2\pi)^{-d}
\lim_{\epsilon\rightarrow 0}
\iint  \, dx\,  dy \, e^{-ixy}  \, \chi(\epsilon x,\epsilon
y) \, b^{\mu}(x,y)
\end{align*}
for $\Psi, \Phi \in \mathcal{D}^{\infty}(dX) :=\{
\varphi \in\mathcal{D}(dX)| U(x)\varphi \in \mathcal{D}(dX)$  is smooth in $\|\cdot\|_{\mathcal{H}}
\} $ and 	   $\chi \in\mathscr{S}(\mathbb{R}^d\times\mathbb{R}^d)$ with $\chi(0,0)=1$. Note that for our case, i.e. for the special algebra that we chose, the stable domain $\mathcal{D}^{\infty}(dX)$ exists and is stable and it is given by $\mathscr{S}(\mathbb{R}^d)$.  In what follows we prove that the oscillatory integral is bounded for the particular algebra in Lemma \ref{ex1}. It is important to stress at this point that we are not interested in the deformation of off-diagonal expressions, i.e.  $(dX^{\mu}dX^{\nu})_{\Theta}$-terms since the Minkowski metric cancels those. Hence, the main object of interest in this section is given by $(dX^{\mu})^2_{{\,\,\Theta }}$.
\begin{lemma}\label{l3}
	The scalar product given by the function $b^{\mu}(x,y)$ (see definition above) is bounded by a real 
	finite constant  $D^{\mu}_{0,\,0}\in\mathbb{R}$, for each $\mu$, and a
	function according to 
	\begin{equation*}
\vert b^{\mu}(x,y)\vert = \vert{\langle \Psi,  
	U(y)\alpha_{\Theta x}(dX^{\mu})^2   \Phi\rangle}\vert\leq 	D^{\mu}_{0,\,0}\, e^{-2a^0(\Theta x)_0},  
	\end{equation*}
 for all $\Psi\in\mathcal{H}$ and $\Phi\in\mathscr{S}(\mathbb{R}^d)$. Let $\gamma$ and $\kappa$ be multi-indices. Then, all derivatives of the function $b^{\mu}(x,y)$ w.r.t. the variables $x$ and $y$ are bounded by the finite constant $D^{\mu}_{\gamma,\,\kappa}\in\mathbb{R}$, for all $\mu$'s, and a
 function as 
 	\begin{equation*}
 	\vert  \partial_{x}^{\gamma}\partial_{y}^{\kappa} b^{\mu}(x,y)\vert = \vert{ \partial_{x}^{\gamma}\partial_{y}^{\kappa}\langle \Psi,  
 		U(y)\alpha_{\Theta x}(dX^{\mu})^2   \Phi\rangle}\vert\leq  D^{\mu}_{\gamma,\,\kappa}\, e^{-2a^0(\Theta x)_0},  
 	\end{equation*}
  for all $\Psi,\,\Phi\in\mathscr{S}(\mathbb{R}^d)$. By using the properties of the cut-off functions and the oscillatory integral it follows that the deformation of the  differentials, i.e. $(dX^{\mu})^2_{\Theta}$,  is  well-defined on $\mathscr{S}(\mathbb{R}^d)$. 
\end{lemma}

\begin{proof}
	In order to prove that the oscillatory integral is bounded, we have to investigate   derivatives of arbitrary order of the  scalar product $b(x,y)$ and prove its boundedness. To achieve this objective we first have to calculate the adjoint action  $\alpha_{\Theta x}$ of the representation of the differential $dX$. Let us recall the algebra,
	\begin{align*}
	[X^0,dX^0]= ia^0dX^0, \qquad\qquad  [X^i,dX^i]= ia^idX^i 
	\end{align*}
	with all other commutator relations being zero. Hence  the explicit adjoint action of the zero component  $dX^0$ is given, by using the Backer Hausdorff Formula, as follows, 
	\begin{align}\nonumber 	\alpha_{\Theta x}(dX^0)&= dX^0+i(\Theta x)_0	[X^0,dX^0]+\cdots\\\nonumber
	&= dX^0+i(\Theta x)_0\,	ia^0dX^0+i^2(\Theta x)_0^2	(ia^0)^2 dX^0+\cdots\\\nonumber
	&= 	dX^0+\sum_{n=1}^{\infty}\frac{i^{2n}}{n!}(a^0)^n(\Theta x)_0^ndX^0\\&\label{ad1}=
	e^{-a^0(\Theta x)_0}dX^0,
	\end{align}
	where in the last lines we used the commutator relations and the notation $(\Theta x)_0=\Theta_{0i}x^i$.
	By using the unitary operators from which the adjoint action is composed of we obtain,
	\begin{align}\nonumber
	\alpha_{\Theta x}(dX^0)^2=(\alpha_{\Theta x}(dX^0))^2=
	e^{-2a^0(\Theta x)_0}(dX^0)^2.
		\end{align}
	 For the spatial components an equivalent calculation can be made and we obtain, 	 
	\begin{align*}
	\alpha_{\Theta x}(dX^i)^2&= 
	e^{-2a_i(\Theta x)_i}(dX^i)^2.
	\end{align*}
	 Turning to the derivatives of   $b^{\mu}(x,y)$ we have for $\mu=0$ the following estimates
		\begin{align*} \partial_{x}^{\gamma}\partial_{y}^{\kappa}
		{\langle \Psi, 
			U(y)\alpha_{\Theta x}(dX^{0})^2   {\Phi}\rangle}&\leq \| (-iX)^{\kappa} \Psi\| \| 
		\partial_{x}^{\gamma}\alpha_{\Theta x}(dX^0)^2 {\Phi}\|\\&\leq 
			\| (-iX)^{\kappa} \Psi\|     \| 
		\partial_{x}^{\gamma}e^{-2a^0(\Theta x)_0}(dX^0)^2 {\Phi}\|\\&\leq
		D^{0}_{\gamma,\,\kappa}\,e^{-2a^0(\Theta x)_0},
		\end{align*}
		where  $\gamma, \kappa$ are multi-indices. In the last lines we  used the explicit adjoint action and the fact that there exists a finite constant $D^{0}_{\gamma,\kappa}$, for all $\gamma$'s and $\kappa$'s, due to the appropriately chosen domains. To prove that the oscillatory integral is finite we use the former inequality, i.e. 
		\begin{align*}
		&(2\pi)^{-d}
		\lim_{\epsilon\rightarrow 0}
		\iint  \, dx\,  dy \, e^{-ixy}  \, \chi(\epsilon x,\epsilon
		y) \,\partial_{x}^{\gamma}\partial_{y}^{\kappa} b^{0}(x,y)
		\\ &\leq(2\pi)^{-d} \,	D^{0}_{\gamma,\,\kappa}
		\lim_{\epsilon\rightarrow 0}
		\iint  \, dx\,  dy \, e^{-ixy}  \, \chi(\epsilon x,\epsilon
		y) \,e^{-2a^0(\Theta x)_0}
		\\&= (2\pi)^{-d} \,	D^{0}_{\gamma,\,\kappa}
		\lim_{\varepsilon_1\rightarrow 0}  \left(
		\int dx \lim_{\varepsilon_2\rightarrow 0} 
		\left(\int dy  e^{-ixy}
		\chi_2(\varepsilon_2 y)\right)\,\chi_1(\varepsilon_1  x)\,e^{-2a^0(\Theta x)_0}
		\right) 
		\\
		&=       (2\pi)^{-d/2}	D^{0}_{\gamma,\,\kappa}
		\lim_{\varepsilon_1\rightarrow 0}  \left(
		\int dx \,
		\delta( {x} )\,\chi_1(\varepsilon_1  x) \,e^{-2a^0(\Theta x)_0} \right) =	D^{0}_{\gamma,\,\kappa},
		\end{align*}
	The oscillatory integral does not depend on the chosen cut-off function. Hence, we choose  the regulator as   $\chi (\varepsilon x,\varepsilon y)= \chi_2(\varepsilon_2 x
		)\chi_1(\varepsilon_1 y)$ 
		with $\chi  \in\mathscr{S}(\mathbb{R}^{d}\times\mathbb{R}^{d})$ and $\chi_{l}(0 )=1$, $l=1$, $2$,
		and we obtain a delta
		distribution $\delta(x)$ in the limit $\varepsilon_2 \rightarrow
		0$, \cite[Section 7.8, Equation 7.8.5]{H}.    Moreover, the resulting integral converges and can be seen by using for example the regulator function $\chi_2(\varepsilon_2 x)=c_1 e^{-\varepsilon_2 x^2}$. Similar considerations can be done for the spatial components.

\end{proof}
 
\begin{lemma}\label{defalg}
	By using the representation of differential algebra (given in Lemma \ref{ex1}) we obtain a well-defined, deformed differential operator $(dX^{\mu})^2_{\,\,\Theta}$ on $\mathscr{S}(\mathbb{R}^d)$ and it is explicitly given by
	\begin{equation}
	(dX^{\mu})^2_{\,\,\Theta}=e^{-2a_{\mu}(\Theta X)_{\mu}}(dX^{\mu})^2.
	\end{equation}
	
\end{lemma}
\begin{proof}
	Since we have proven that deforming the differential structure is a well defined expression, we use the spectral representation in order to obtain the result of deformation. Hence for $\mu=0$ we have the following expression
	
		\begin{align*} (dX^{0})^2_{\,\,\Theta}&=
		\int dE(x) \, \alpha_{\Theta x}(   dX ^{0})^2  \\&=
			\int dE(x) \, e^{-2a^0(\Theta x)_0}(   dX ^{0})^2	\\&=
			e^{-2a^0(\Theta X)_0}(   dX ^{0})^2 ,
		\end{align*} 
		where we used the explicit adjoint action (\ref{ad1}). For the spatial components, i.e. $\mu=i$, we have 
		\begin{align*} (   dX ^{i})^2_{\,\,\Theta}&= 
		\int dE(x) \, \alpha_{\Theta x}(   dX ^{i})^2  \\&=
		\int dE(x) \, e^{-2a_i(\Theta x)_i}(   dX ^{i})^2 	\\&=
		e^{-2a_i(\Theta X)_i}(   dX ^{i})^2 .
		\end{align*} 
\end{proof} 
\subsection{Self-Adjointness}
The flat metric consists of one-forms that are represented as self-adjoint operators. In this context it is important to investigate if the property of self-adjointness is kept post deformation. From a quantum mechanical point of view,   the re-expression of the metric as a self-adjoint entity has to be taken. The reason therein lies in the fact that we want to connect geometrical expressions with quantum mechanical  (QM) observables. Quantum mechanical observables, on the other hand,  have to be self-adjoint in order have real eigenvalues.  Hence, in this section we investigate if the deformed one-forms, i.e. the deformed differentials, are self-adjoint. 
\\\\
In order to investigate self-adjointness of the deformed differentials we give the following Lemma.  
\begin{lemma}\label{wcl2}
	Let $\Theta$ be a real skew symmetric matrix on $\mathbb{R}^d$, ${A}$   a densely defined operator on a Hilbert space $\mathscr{H}$ such that the deformation, i.e. $A_{\Theta}$, is well-defined on a dense domain $\mathcal{D}$.
	Then the following relations hold  on  $\mathcal{D}$, \\ \begin{enumerate} \renewcommand{\labelenumi}{(\roman{enumi})}
		\item   $\int \alpha_{\Theta x}(A)dE(x)=\int dE(x)\alpha_{\Theta x}(A)$
		\item $\left(\int \alpha_{\Theta x}(A)dE(x)\right)^{*}\subset\int \alpha_{\Theta
			x}(A^{*})dE(x)$
	\end{enumerate}
\end{lemma} 
\begin{proof}
	For a subset of bounded operators, namely operators belonging to a $C^*$-algebra which are smooth w.r.t. the adjoint action, the proof can be found in \cite{BLS}. However,  for a particular  subset of unbounded operators an equivalent statement holds. Since, we assume that the deformed (unbounded) operators are well-defined item $(i)$, follows by the same proof as in \cite{BLS}. In particular, the proof is done by expressing the unitary operators and the spectral measure in terms of strong limits. This is all well-defined since   the expressions are, for this particular case, all bounded on the  dense domain $\mathcal{D}$. Hence, the proof  is analog and equivalent statements hold as in \cite{BLS}.
	To prove item $(ii)$ we use as in the the original work item $(i)$.
\end{proof}
By using the former results we obtain the following.
\begin{lemma}\label{l5}
	The deformed one-forms, that are given in Lemma \ref{defalg} according to
		\begin{equation}
		(dX^{\mu})^2_{\,\,\Theta}=e^{-2a_{\mu}(\Theta X)_{\mu}}(dX^{\mu})^2,
		\end{equation}
	 are self-adjoint on the dense domain $\mathscr{S}(\mathbb{R}^d)$.
\end{lemma}
\begin{proof}
	Symmetry of the deformed differentials follows from Lemma \ref{wcl2}, $(ii)$ since we started with self-adjoint operators.	The  proof for self-adjointness is done by using the properties of the representations  and the  dense domain $\mathscr{S}(\mathbb{R}^d)$. In particular, the vectors of the Schwartz space are analytic w.r.t. the representations of infinitesimal generators of the Heisenberg-Weyl group. Hence, the deformed symmetric operator has a total set of analytic vectors from which self-adjointness follows, \cite[Theorem X.39]{RS2}.  
\end{proof}
As already explained in the beginning of this Section, the terms emerging from deformation are pulled to the flat metric. Hence, the new metric that emerged from deformation is by the same arguments as before, self-adjoint.  
\subsection{Uniqueness}
The reader familiar with Rieffel deformations and warped convolutions might have, rightfully, the  question of uniqueness in  mind. In particular, by uniqueness, we mean  the possibility of defining the deformation as the product of two deformed differentials instead of taking the deformation of the product of   differentials, which, in general, is different. 
In this context we introduce the well-known  Rieffel product (see \cite{R}) in order to prove uniqueness of our deformation scheme in our notation (see \cite{BLS} as well). 
\begin{lemma}\label{l2.1}
	Let $\Theta$ be a real skew-symmetric matrix on $\mathbb{R}^d$ and let $  {A},   {B}$ be  
	densely defined operators on a Hilbert space $\mathscr{H}$ such that their respective deformations are well-defined on a dense domain $\mathcal{D}$. Then,
	\begin{equation*}
	{A}_{\Theta}   {B}_{\Theta} = (A\times_{\Theta }B)_{\Theta} , 
	\end{equation*}
	where $\times_{\Theta}$ is the \textbf{Rieffel product} defined by 
	\begin{equation}\label{dp0}
	(A\times_{\Theta}B ) =(2\pi)^{-d}
	\lim_{\epsilon\rightarrow 0}\iint dx\, dy \,\chi(\epsilon x,\epsilon y )\,e^{-ixy} \, \alpha_{\Theta x}(A)\alpha_{y}(B) .
	\end{equation}
\end{lemma} $\,$\\
Hence,  there are two expressions  worth investigating, i.e. \newline
\begin{itemize} 
	\item    $(dX^{\mu}dX^{\mu})_{\Theta}$\newline
	\item  $dX^{\mu}_{\Theta}dX^{\mu}_{\Theta}=(dX^{\mu}\times_{\Theta} dX^{\mu})_{\Theta}$, 
\end{itemize} $\,$ \newline 
where the deformed product, $\times_{\Theta}$, appears in the second option (see Formula  \ref{dp0}). Since, we investigated the first case and received the well-defined outcome, we investigate in the forthcoming part the second definition.  
\newline 
\begin{proposition}
	The differentials deformed with warped convolutions, i.e. $(dX^{\mu})^2_{\,\,\Theta}$ ,   are   regardless of the possible definitions unique, i.e. the following equivalence holds,
	\begin{align*} 
	(dX^{\mu}dX^{\mu})_{\Theta}=(dX^{\mu}\times_{\Theta} dX^{\mu})_{\Theta}.
	\end{align*} 
	  In particular, on the domain  $\mathscr{S}(\mathbb{R}^d)$, the Rieffel product is equal to the undeformed product,
	\begin{align*} 
	dX^{\mu}\times_{\Theta} dX^{\mu}=(	dX^{\mu} )^2.
	\end{align*}  
\end{proposition}
\begin{proof}
	We start the proof by using the definition of the deformed product (see Formula \ref{l2.1}) for $\mu=0$, i.e. 
	
	\begin{align*} &
	dX^{0}\times_{\Theta} dX^{0}=(2\pi)^{-d}
	\lim_{\epsilon\rightarrow 0}\iint dx\, dy \,\chi(\epsilon x,\epsilon y )\,e^{-ixy} \, \alpha_{\Theta x}(	dX^{0})\alpha_{y}(dX^{0}) \\&=
	(2\pi)^{-d}
	\lim_{\epsilon\rightarrow 0}\iint dx\, dy \,\chi(\epsilon x,\epsilon y )\,e^{-ixy} \,  
	e^{-a^0(\Theta x+y)_0}	(dX^{0})^2\\&=(2\pi)^{-d}
	\lim_{\epsilon\rightarrow 0}\iint dx\, dy \,\chi_3(\epsilon x,\epsilon y )\,e^{-ixy} \,  
	e^{-a^0(\Theta x)_0}	(dX^{0})^2
	\\&=	(2\pi)^{-d}
	\lim_{\varepsilon_1\rightarrow 0}  \left(
	\int dx \lim_{\varepsilon_2\rightarrow 0} 
	\left(\int dy \, e^{-ixy}
	\chi_2(\varepsilon_2 y)\right)\,\chi_1(\varepsilon_1  x)	e^{-a^0(\Theta x)_0}	(dX^{0})^2\right)
	\\&=	 (2\pi)^{-d/2}		\lim_{\varepsilon_1\rightarrow 0}  \left(			
	\int dx \,
	\delta( {x} )\,\chi_1(\varepsilon_1  x) \,e^{-a^0(\Theta x)_0} (dX^{0})^2\right)=(dX^{0})^2,
	\end{align*}
	where in the last lines we used the explicit adjoint action (see Equation \ref{ad1}), we performed a variable substitution $(x,y)\rightarrow(x-\Theta^{-1}y,y)$ and furthermore we chose the regulators as in the proof of Lemma \ref{l3}. The proof for the spatial part is analogous. 
\end{proof}
Since we have shown in the former Lemma that the deformed product is equal to the undeformed, our two possible definitions agree, and hence uniqueness of the deformation follows. To clear the question of uniquenesses is, apart from mathematical curiosity,  important with   regards to the physical outcome.

\section{Physical Outcome}
The last sections were devoted to the study of the rigorous state of the strict deformation. We proved that the deformation is well-defined and it is given by self-adjoint operators. In this section we study the outcome, i.e. the results, of the deformation. In particular, we obtain a curved space-time metric which we relate to well-known physical models.  
	\subsection{Family of Conformal-Flat Space-Times}
This section is devoted to the class of space-times that we obtain by deformation. 
In particular, there are two deformations that lead to an interesting space-time. First of all, we introduced the deformation of the commutative differential structure by a freedom that exists due to the commutative commutator relations of the algebra. The second step in direction of the following result, was the strict Rieffel-deformation. 
\begin{theorem}\label{t1}
	The deformed differential algebra, given in Lemma \ref{ex1}, gives the following warped convoluted line-element 
	\begin{align*} 
	(ds^2)_{\Theta}= (\eta_{\mu\nu})_{\Theta}\, 	d\hat x^{\mu}d\hat x^{\nu},
	\end{align*}
where from the deformation of the flat   line-element we obtain the \textbf{curved space-time metric} $(\eta_{\mu\nu})_{\Theta}=e^{-2a_{\mu}(\Theta \hat x)_{\mu}}\eta_{\mu\nu}$. 
\end{theorem}
	\begin{proof}
		The deformed line-element is obtained by using the multiplicative properties of the faithful representation and is essentially the result of Lemma \ref{defalg}. From the deformation we obtained an $\hat x$-dependent conformal factor which can be included in the metric and hence it results in a curved space-time.
	\end{proof}
	By applying two deformations  a curved space-time was generated out of a flat one. In particular, we obtain a whole class of conformal flat space-times depending on the choice of the deformation parameters. 
Moreover, the theorem represents a new path to curving space-time by a strict  deformation procedure.  Namely, the deformation acts as a gravitational field, i.e. a source that is curving the space-time.   
   	 \subsection{Ultra-Static Space-Times}
   	 An interesting example that we obtain  from the simplest choice of our algebra (see Equation \ref{rd1}) and the deformation, is that of an ultra-static space-time. The line-element of such a space-time is given by, \cite[Definition 3.1.2.]{DAP1} or \cite[Chapter 6.1]{WA},
   	 \begin{equation}\label{us}
   	 ds^2=dt^2-h_{ij}(x,y,z)dx^i\,dx^j.
   	 \end{equation}
   	 These space-times are particular interesting with regards to the geometry they generate (see for example \cite{SH} and \cite{SST}). For instance, on such space-times the paths of light coincide with the geodesics of the spatial part of the metric. Moreover, they are the only class of space-times that admit covariant constant time-like vector fields. Since this admission means that such vector fields do not accelerate, they   fall under a rightful extension of  Minkowskian inertial frames. 
   	 \\\\
   	 Apart from their physical interest in regards to gravity, they have as well generated a lot of interest from an algebraic quantum field theoretical (AQFT) point of view. This stems from the fact that for such space-times, the    strong energy nuclearity condition  has been proven for the   free massive Klein-Gordon field, \cite{V}. For further  applications in AQFT see \cite{SA1}, \cite{FS}, \cite{GWU},  \cite{GS} and \cite{SK},    to mention just a few. \\\\
   	 Next, we turn our attention to the deformation and   outcome of an ultra-static space-time given by the following theorem. 
   	 
   	 \begin{theorem} 
   	 	Let the  deformation matrix $\Theta_{0j}=0$ be equal to zero. Then,	the deformed line-element, given in Theorem \ref{t1}, admits  an \textbf{ultra-static space-time}. The line-element that is curved by deformation  is given by
   	 	\begin{align*} 
   	 	(ds^2)_{\Theta}=   d\hat t^2-h_{ij}(\hat{\mathbf{ {x}}})\,d\hat{x}^{i}d\hat{x}^{j},
   	 	\end{align*}
   	 	where $\hat x^0=\hat t$ and the spatial metric $h_{ij}(\hat{\mathbf{ {x}}})$ has the following form

   	 	\begin{align}
   	 	h_{ij}(\hat{x})=	\begin{pmatrix} e^{-2a_1(\Theta \hat{x})_1} &  0 & 0 \\ 0 & e^{-2a_2(\Theta\hat{x})_2}& 0   \\ 0 & 0 &e^{-2a_3(\Theta\hat{x})_3}
   	 	\end{pmatrix}.
   	 	\end{align}
   	 \end{theorem}
   	 \begin{proof}
   	 	The proof follows by inserting the particular choice of deformation matrix, i.e. $\Theta_{0j}=0$, for the metric  given in Theorem \ref{t1} and  from the definition of an ultra-static space-time, see Equation (\ref{us}).
   	 \end{proof}
   	 
   	 Since it is a   specific ultra-static space-time, this result may seem   restrictive. However, it is possible to obtain a more complex and hence a more generic spatial metric $h_{ij}(\hat{x})$ by choosing as the generators of deformation real-valued functions of the coordinate operators $\hat{x}$. Where, here, we refer specially to representations of the respective operators. Nevertheless, a reevaluation of the proofs concerning convergence and self-adjointness has to be done. In particular, the respective domains have to be chosen such that the whole procedure remains in a strict framework. 
		\subsection{Friedmann-Robertson-Walker}
		In this Subsection we investigate a specific physical example which results from the simplest deformation in four dimensions. Since, we know from the last section that we obtain a space-time with a  coordinate depending conformal term, the Friedmann-Robertson-Walker space-time  is in reach. 
	\begin{theorem}
Let the spatial parameters ${a}$ and the deformation matrix $\Theta$ be given as, 
\begin{align*}
a_0\Theta_{0i}\approx 0,\qquad  2a_i\Theta_{0i}=-H ,\qquad \Theta_{ij}=0, 
\end{align*} 
where $H$ is the Hubble parameter. Then, in four dimensions, the warp convoluted line-element from Theorem \ref{t1} gives the \textbf{Friedmann-Robertson-Walker space-time}. The flat line-element that is curved by deformation  is given by
		\begin{align*} 
		(ds^2)_{\Theta}&=   d\hat t^2-e^{ H\hat t} d\hat{\mathbf{x}}^2,
		\end{align*}
		where $\hat x^0=\hat t$ and $(d\hat{\mathbf{x}})^2=d \hat x_1^2+d\hat x_2^2+d\hat x_3^2$. Hence, we obtain the following \textbf{deformed flat space-time metric}, 
		  	\begin{align*}
(	\eta_{\mu\nu})_{\Theta} =	\begin{pmatrix}  1 &  0 & 0 & 0 \\ 0 & -e^{ H\hat t}& 0  & 0 \\ 0 & 0 &-e^{ H\hat t}& 0\\ 0 & 0 &0& -e^{ H\hat t}
		  \end{pmatrix}.
		  \end{align*}
 
	\end{theorem}

	\begin{proof}
		By using the former result (Theorem \ref{t1}) and the specific choice of deformation parameters the theorem   can be easily proven.   In order for the representation of the differential algebra to be faithful we demanded that the parameters $a$ are non-vanishing. 
	\end{proof}
 
	In the former theorem we obtained an interesting result concerning the FRW space-time. It can be obtained by the deformation of the differential structure and by a well-defined procedure as the Rieffel quantization. 		Since, we originally had additional terms before applying the condition $a_0\Theta_{0i}\approx 0$ there are corrections to the usual FRW-cosmology. These deformations, if physically relevant, have to lead to meaningful results.  In the following, we would like to investigate the implications more thoroughly. In particular, it is interesting to investigate the scenario where we start with a  FRW-type metric, apply the deformations on it and try to solve the classical Einstein-equations for the deformed metric. Hence, are there solutions or in particular how do the well-known solutions change by terms appearing from deformation?
	\begin{lemma}
		Let the metric the line-element  be given as 
				\begin{align}  
		 ds^2 &=   d\hat t^2-\tilde{a}(\hat{t})^2 d\hat{\mathbf{x}}^2		\end{align}
	 Then, the deformation, for the  choice $\Theta_{ij}=0$, induces the deformed metric  given as 
	 	\begin{align}  ds^2 &=   e^{-2b_ix^{i} } d\hat t^2- {a}(\hat{t})^2 d\hat{\mathbf{x}}^2	 ,
	 	\end{align}
	 	where we defined the new variables $a (\hat{t})^2:=e^{-2a_{i} \Theta_{i0} \hat t  } \tilde{a}(\hat{t})^2$, with $b_{i}=a_{0}\Theta_{0i}$.
	\end{lemma}

	\begin{proof}
	To see that the deformation of the curved metric induces the terms that are written above we simply use the former result (Theorem \ref{t1})  and the commutativity of the coordinate operators $\hat{x}$.
	\end{proof} 
In what follows we calculate the Einstein Equations for the deformed Friedmann-Robertson-Walker metric. In order to do so we assume that our obtained space-time is classical, in the sense that standard Riemann geometry applies and we can use standard techniques to calculate the curvature. Hence, we pose the question if the space-time that we obtained by deformation quantization has classical solutions of the Einstein Equations as well. The metric components are given by, 
\begin{equation}
g_{00}=  e^{-2b_ix^{i} },\qquad g_{0j}=0,\qquad g_{ij}=-\delta_{ij}a^2(t)
\end{equation} 
For such a metric we have the following   Einstein equations,    
   \begin{align} \label{e1}
   3     \frac{\dot a(t)^2}{a(t)^2}  +\Lambda g_{00}&=\kappa \rho \\
   -2 \frac{\dot a(t)}{a(t)}b_i&=\kappa T_{0i}\\\label{e2}
   2 \frac{\ddot a(t) }{a(t) }+ \frac{\dot a(t)^2}{a(t)^2}  +\frac{2}{3}   b^2  g_{00}    + \Lambda  g_{00}& =-\kappa P\,g_{00},
   \end{align}  
   Let us for the moment set $\Lambda=0$. The second equation does not imply in principle anything new except for the fact that $T_{0i}$ comes with a constant depending on the deformation parameter $b$, hence in the undeformed case it is equal zero. 
   We proceed as in \cite[Chapter 5.2]{WA}  where we consider the case   $\Lambda=0$ and $b_{i}=\frac{1}{2}\Theta e_{i}$, where $e_{i}$ is a unit vector. Next we insert Equation (\ref{e1}) into Formula (\ref{e2}) and we obtain, 
   \begin{align*} 
   3\frac{\ddot a(t) }{a(t)} =-4\pi\left(\rho+3 P g_{00}
   \right)-2c^2g_{00}.
   \end{align*} 
   Since we are not interested in the most general solution but rather investigate the question whether there exists one at all for the deformed case, for $\Theta\neq 0$, we assume dust like matter $P=0$ and moreover, we only take perturbations to second order in the deformation parameter $\Theta$.  Hence, we end up with two equations, where we choose the units such that $\kappa=8\pi$,
    \begin{align} \label{e3}
      3     \frac{\dot a(t)^2}{a(t)^2}   &=8\pi\rho\\\label{e4}
    3\frac{\ddot a(t) }{a(t)} &=-4\pi \rho  +\frac{3}{2} \Theta^2 a(t)^{-2}+\mathcal{O}(\Theta^3),
    \end{align} 
   By multiplying Equation (\ref{e1}) with $a^2$ and taking the time-derivative w.r.t. the first equation and insert for the function $\ddot a(t)$ the second equation we obtain, 
 \begin{align} 
\dot{\rho}a^3=-3\rho\,\dot{a}\,a^2+ \frac{3}{8\pi}\Theta^2 \dot{a} ,
 \end{align}
 which   rewritten as $\frac{d}{dt}({\rho}a^3)=  \frac{3}{8\pi}\Theta^2 \dot{a}$
reads
  \begin{align*} 
  {\rho}\,a(t)^3 =  \frac{3}{8\pi}\Theta^2  {a}(t) +C^{'},\qquad C^{'}\in\mathbb{R}.
  \end{align*}
 Next, by multiplying Equation (\ref{e3}) by the   function $a(t)^3$ we have,  
   with $C= \frac{8 }{3}\pi C^{'}$  
      \begin{align*} 
      \dot a(t)^2     -\frac{C}{a(t)} -  \Theta^2 =0.
      \end{align*}
 This result is equal to dust-like matter in the Friedmann-Robertson-Walker universe with the non-commutative parameter $\Theta$ playing the role of  $k$, where $k$ belongs to the set $\{-1,0,1\}$ and plays  the role of the spatial curvature index. However, since the square of the real deformation constant $\Theta$ is positive this means that the deformation interplays between a flat and a negatively curved spatial sector. Hence, in the former results we obtained  interesting results concerning the FRW space-time.   In particular we studied the additional coordinate dependent effect induced by deformations and obtained physically interesting solutions. This is the most effective scenario in terms of what deformation quantizations can achieve. Starting from the simplest theory they have to, apart from inducing the well-known models,  induce new ones. Ideally, as in our case, these new models have, often enough, a physical significance.   
\section{Moyal-Weyl as a Consistency Condition}
In this section we study the implication of non-commutativity between space and time on the conformal-flat space-times that were induced from a flat metric by deformations. Hence, we investigate how this specific space-time behaves  under the additional assumption of a non-commutative space-time. We start by assuming that we have a Moyal-Weyl type of non-commutativity between space and time,
      \begin{align*} 
[\hat  x^{0},\hat x^{j}]= i\,\Omega^{0j},\qquad\Omega^{0j}\in \mathbb{R}^n
      \end{align*}
      with all other commutators being equal to zero. The representations of these operators are as before, with the addition of a one  higher dimensional Heisenberg-Lie algebra and we denote the representations by $\pi(\hat x^{\mu})=Y^{\mu}$ and $\pi( d\hat x^{\mu})=dY^{\mu}$,    
           \begin{align*} 
 Y^{0}= X^{0}+Q^{d+1} ,\qquad    Y^{i}= X^{i}+ \Omega^{0i}P^{d+1}
           \end{align*}
      The representation of the differentials is given as before by
     $$ \pi(d\hat{x}^{\mu}) = i{q}_{b}\,[P^{b} , \pi(\hat{x}^{\mu})],\qquad b=1,\dots,d,$$  and the well-definedness of the deformation holds as before on the dense and stable domain $\mathscr{S}(\mathbb{R}^{d+1})$.   
     \\\\
In non-commutative differential geometry, one can build a metric out of a given associative algebra and calculate the associated curvature tensors and the Ricci scalar. However, in addition to the, from commutative differential geometry known, requirements such as metricity and torsion freeness the condition of centrality of the metric is essential, see \cite{BM}. In particular, only for the case of central bi-modules (see \cite[Chapter 6-9]{L}, \cite{DV1}, \cite{DV2} and references therein) the framework of non-commutative geometry induces the non-commutative analogue of geometric quantities.  In the case at hand centrality for the metric tensor is given by,  
    \begin{align*} 
    [ \hat{x}^{\mu},g_{\nu\rho}d\hat{x}^{\nu}\otimes d\hat{x}^{\rho}]= 0,
    \end{align*}
   where a summation   which can be rewritten as,  
       \begin{align*} 
       [ \hat{x}^{\mu},g_{\nu\rho}]d\hat{x}^{\nu}\otimes_{\mathcal{A}_c} d\hat{x}^{\rho} +
   g_{\nu\rho} [ \hat{x}^{\mu},d\hat{x}^{\nu}]\otimes_{\mathcal{A}_c} d\hat{x}^{\rho} + g_{\nu\rho}d\hat{x}^{\nu}\otimes_{\mathcal{A}_c} [ \hat{x}_{\mu}, d\hat{x}^{\rho}]    = 0\\
    [ \hat{x}^{\mu},g_{\nu\rho}]d\hat{x}^{\nu}\otimes_{\mathcal{A}_c} d\hat{x}^{\rho} +
  C^{\mu\nu}_{\lambda}  \, g_{\nu\rho} d\hat{x}^{\lambda}\otimes_{\mathcal{A}_c} d\hat{x}^{\rho} +C^{\mu\rho}_{\lambda}\, g_{\nu\rho}d\hat{x}^{\nu}\otimes_{\mathcal{A}_c}   d\hat{x}^{\lambda} = 0
  \\
\left(  [ \hat{x}^{\mu},g_{\nu\rho}]  +
  C^{\mu\kappa}_{\,\,\,\,\nu}  \, g_{\kappa\rho}   +C^{\mu\kappa}_{\,\,\,\,\rho}\, g_{\nu\kappa}\right)d\hat{x}^{\nu}\otimes_{\mathcal{A}_c}   d\hat{x}^{\rho} = 0,
       \end{align*}
    which leads to the following commutator relations. 
     \begin{align*} 
   [ \hat{x}^{\mu},g_{\nu\rho}]  =-\left(
     C^{\mu\kappa}_{\,\,\,\,\nu}  \, g_{\rho\kappa}  +C^{\mu\kappa}_{\,\,\,\,\rho}\, g_{\nu\kappa}  \right).
     \end{align*}
    The fact that the metric only depends on the   generators of the algebra, suggests that if the algebra is chosen commutative the sum on the right has to be zero. For the case of the FRW-space-time of inflation and for the specific choice of the $C's\neq0$ it does not hold! However, since we turned on the non-commutativity of space and time   solutions to this problem exist.
    
    \begin{theorem}
  Let the deformation matrix be given by  $\Theta_{0j}=\Theta e_{j}$, where $\Theta$ is a real parameter and $e$ is the unit vector.    Then, for  the class of conformal flat metrics $g_{\mu\nu}=(\eta_{\mu\nu})_{\Theta}$, that are obtained by deformation, with $\Theta_{ij}=0$, \textbf{centrality} is fulfilled by assuming the constant non-commutativity of the Moyal-Weyl plane
   $$ [\hat  x^{0},\hat x^{j}]= i\,\Omega^{0j}=i\Omega e^j $$  with the additional condition
    $  \Omega^{-1} =  {n}\Theta$.
    \end{theorem}

    \begin{proof} 
     For our concrete case, i.e. for the class of conformaly flat metrics with $\Theta_{ij}=0$ we have, 
     \begin{align} \label{e6}
    \eta_{\nu\rho} [ \hat{x}^{\mu},e^{-2a_{\nu}(\Theta \hat x)_{\nu}}]  =-\left(
     C^{\mu\kappa}_{\,\,\,\,\nu}  \, g_{\rho\kappa}  +C^{\mu\kappa}_{\,\,\,\,\rho}\, g_{\nu\kappa}  \right).
     \end{align}
Moreover, for simplicity we prove it in two dimensions. 
     By taking $\nu=0, \rho=0$ we obtain     
     \begin{align} \label{e7}
 [ \hat{x}^{\mu},e^{-2a_{0}(\Theta \hat x)_{0}}]  =-2 
     C^{\mu\kappa}_{\,\,\,\,0}  \, \eta_{0\kappa}  e^{-2a_{\kappa}(\Theta \hat x)_{\kappa}}. 
     \end{align}
    If $\mu$ is a spatial component the left and the right-hand vanish. Hence, we are left with the case $\mu=0$,
       \begin{align} \label{e5}
       [ \hat{x}^{0 },e^{-2a_{0}(\Theta \hat x)_{0}}]  =-2 
       C^{0\kappa}_{\,\,\,\,0}  \, \eta_{0\kappa}  e^{-2a_{\kappa}(\Theta \hat x)_{\kappa}}  =-2 
      ia_0 \,    e^{-2a_{0}(\Theta \hat x)_{0}}  ,
         \end{align}
    where in the last line we used the specific form of the $C$'s, i.e. $C^{00}_{\,\,\,\,0}=ia_0$, $C^{11}_{\,\,\,\,1}=ia^1$, with all others being zero. The left-hand side is calculated by using the Taylor-expansion of the exponential function, 
     \begin{align*} 
     [ \hat{x}^{0 },e^{-2a_{0}(\Theta \hat x)_{0}}] &  =\sum_{k=0}^{\infty}  \frac{(-2a_{0}\Theta_{01})^{n} }{n!} [ \hat{x}^{0 },( \hat x^{1})^{n}]
     \\&=\sum_{k=0}^{\infty} \frac{(-2a_{0}\Theta_{01})^{n} }{n!} n\,  [ \hat{x}^{0 }, \hat x^{1}] (\hat x^{1})^{n-1} \\&
     =   i\Omega^{01}\sum_{k=0}^{\infty} \frac{(-2a_{0}\Theta_{01})^{n} }{(n-1)!}\,   (\hat x^{1})^{n-1} 
     \\&= -2ia_{0}\,\Omega^{01} \, \Theta_{01} \, e^{-2a_{0}(\Theta \hat x)_{0}}
     \end{align*}
    By equating both sides of Equation (\ref{e5}) we obtain the condition 
    $\left( \Omega^{01}\right)^{-1}  = \Theta_{01}$. The case $\nu=0, \rho=0, \mu=1$ vanishes on both sides of Equation (\ref{e7}).  Next, for  $\nu=0, \rho=1$ we have zero on the left-hand side of Equation  (\ref{e6}), and the right-hand side vanishes as well   for all $\mu$'s. Hence, we are left with the case $\nu=1, \rho=1$
        \begin{align*} \label{e9}
  -      [ \hat{x}^{\mu},e^{-2a_{1}(\Theta \hat x)_{1}}]  =-\left(
        C^{\mu\kappa}_{\,\,\,\,1}  \, g_{1\kappa}  +C^{\mu\kappa}_{\,\,\,\,1}\, g_{1\kappa}  \right),
        \end{align*}
     which for $\mu=0$ obviously vanishes, and for $\mu=1$ and the same techniques as before  we have
       \begin{align*}  
    2ia_1\Omega^{01}\Theta_{10}=          [ \hat{x}^{1},e^{-2a_{1}(\Theta \hat x)_{1}}] & =-\left(
       C^{11}_{\,\,\,\,1}  \, g_{11}  +C^{11}_{\,\,\,\,1}\, g_{11}  \right)\\&=+ 2ia_{1}
       e^{-2a_{1}(\Theta \hat x)_{1}},
       \end{align*}  
     which gives us  the same requirement as before.      The calculation for higher dimensions is straight-forward. 
         \end{proof}
     Hence, for the class of conformal flat metrics $(\eta_{\mu\nu})_{\Theta}$, that are obtained by deformation, the non-commutative plane is essential for the requirement of centrality to hold.
In particular, apart from the "geometrical"  measurement problem and the Landau-quantization, see \cite{AA} and \cite{Muc1},   a constant non-commutative space-time such as the Moyal-Weyl is a \textbf{necessity} from a mathematical and  physical point of view and has to be   understood as a consistency requirement. This is a strict mathematical restriction that has far reaching physical implications. The implication, is the following. If we express the metric in terms of operators on a Hilbert space  then for metrics such as the conformal flat families that emerged out of deformation,   \textbf{quantization of space-time is an essential ingredient}. 
 
	\section{Discussion and Outlook}
In this work we have shown that from a strict deformation quantization of flat space-time we were able to obtain a curved space-time. In particular, we investigated how to relate such a deformation with a physical space-time, such as the FRW space-time of a dark energy dominated universe. Moreover, the terms we obtained by deformation turned out to supply us with additional solutions to the classical Einstein-equations. Apart from a coordinate dependent deformation of the well-known flat FRW space-time we obtain in addition a whole class of conformal-flat space-times.   \\\\
Since we chose a particularly simple algebra from the start of our investigation, this scheme has a rich family of possible extensions. For example a more complex algebra  that satisfies the Jacobi identities is  given by, 
\begin{align*}
	[X^0,dX^0]&= i(a\,dX^0+e\,dX),   &[X^0,dX ]&= i(r\,dX^0+f\,dX),\\
	 [X ,dX^0]&= i(r\,dX^0+f\,dX) ,   & [X ,dX  ]&= i(s\,dX^0+h\,dX), 
\end{align*}
with all other commutator relations being zero and	where the constants $a,\,e,\,f,\,h,\,r,\,s$ are real and fulfill the following equations,
	\begin{align*}
 es-rf=0,\qquad f(a-f)+e(h-r)=0,\qquad r(r-h)+s(f-a)=0. 
	\end{align*} Hence, the deformation scheme, in particular the result, is not limited to conformal flat metrics. This result is solely owed to the simple choice of the algebra that we picked in this context. Therefore, more complex structural choices of the algebra, in particular of the constants $C$, result in more general curved space-times. 
	In addition to choosing a different algebra the generators themselves can be taken (on appropriate domains) to be real-valued unbounded functions of the  algebra. This is work in progress. \\\\  
  In particular, what the investigation showed was the following; we expressed the main geometrical element, namely the flat metric, as an operator-valued entity on a Hilbert space. We performed a strict deformation quantization and obtained a curved space-time and hence we had a metric inducing gravity in terms of a well-defined self-adjoint operator-valued tensor. However, for this tensor to be central, hence to satisfy a mathematical condition that is necessary to proceed in the context of non-commutative geometry, 
a quantization of space-time had to be imposed. Therefore,   a non-commutative space-time such as the Moyal-Weyl has to be understood as consistency condition from a mathematical and physical point of view w.r.t. the conformal flat metrics that are expressed in terms of operator-valued tensor products in the before indicated way. For these particular class of operator valued metrics,  quantization of the space-time, i.e. the replacement of space-time by a non-commutative structure such as the Moyal-Weyl is not a luxury but a necessity. 
	\section*{Acknowledgments} 
	We would like to thank Prof. M. Rosenbaum, Prof. D. Vergara , Prof. D. Bahns and Dr. Benito Juarez for many important discussions.  Furthermore, we express our gratitude towards Dr.  A. Andersson, Dr. Z. Much and D. Vidal-Cruzprieto  for an extensive proofreading.

\bibliographystyle{alpha}
\bibliography{allliterature1}

\end{document}